\documentclass[11pt]{article}

\usepackage{graphicx}
\usepackage{amsmath}
\usepackage{amsfonts}
\usepackage{amssymb}
\usepackage{amsthm}
\usepackage[all]{xy}

\newtheorem{lemma}{Lemma}[section]
\newtheorem{cor}[lemma]{Corollary}

\newtheorem{observation}[lemma]{Observation}
\newtheorem{thm}[lemma]{Theorem}
\theoremstyle{definition}
\newtheorem{defi}[lemma]{Definition}
\newtheorem{rmk}[lemma]{Remark}

\begin{document}
\parindent=2em

\title{Isomorphisms of Algebraic Number Fields}
\author{Mark van Hoeij and Vivek Pal\footnote{Florida State University}}
\date{}
\maketitle

\begin{abstract}
Let $\mathbb{Q}(\alpha)$ and $\mathbb{Q}(\beta)$ be algebraic number fields. We describe a new method to find (if they exist) all isomorphisms, $\mathbb{Q}(\beta) \rightarrow \mathbb{Q}(\alpha)$. The algorithm is particularly efficient if the number of isomorphisms is one. 

\end{abstract}

\section{Introduction}

Let $\mathbb{Q}(\alpha)$ and $\mathbb{Q}(\beta)$ be two number fields, given by the minimal polynomials $f(x)=\sum_{i=0}^n{f_i x^i}$ and $g(x) = \sum_{i=0}^n {g_i x^i}$ of $\alpha$ and $\beta$ respectively. In this paper we give an algorithm to compute the isomorphisms $\mathbb{Q}(\beta) \rightarrow \mathbb{Q}(\alpha)$.  Suppose there is an isomorphism then we have the following diagram of field extensions: 

$$
\xymatrix{
\mathbb{Q}(\beta) \ar@{-}[d]^{g(x)} \ar[r]^{\cong} &\mathbb{Q}(\alpha)\ar@{-}[d]^{f(x)}\\
\mathbb{Q}  &\mathbb{Q}}
$$

\par To represent such an isomorphism we need to give the image of $\beta$ in $\mathbb{Q}(\alpha)$, in other words, we need to give a root of $g(x)$ in $\mathbb{Q}(\alpha)$. 
\par We now describe two common methods of computing isomorphisms of number fields. 

\begin{enumerate} 

\item[Method I.] Field Isomorphism Using Polynomial Factorization [11, Algorithm 4.5.6]
\begin{itemize}
\item Find all roots of $g$ in $\mathbb{Q}(\alpha)$. Each corresponds to an isomorphism $\mathbb{Q}(\beta) \rightarrow \mathbb{Q}(\alpha)$. The roots can be found by factoring $g$ over $\mathbb{Q}(\alpha)$. 
	\begin{enumerate}
	\item If done with Trager's method, one ends up factoring a polynomial in $\mathbb{Q}[x]$ of degree $n^2$. 
	\item An alternative is Belabas' algorithm for factoring in $\mathbb{Q}(\alpha)$.
	\end{enumerate}

\end{itemize}

\item[Method II.] Field Isomorphism Using Linear Algebra [11, Algorithm 4.5.1/4.5.5]
\begin{enumerate}
\item Let $\alpha_1, \dots, \alpha_d$ be the roots of $f$ in $\mathbb{Q}_p$ (choose $p$ with $d>0$).
\item Let $\beta_1, \dots, \beta_d$ be the roots of $g$ in $\mathbb{Q}_p$.
\item If $\beta \mapsto h(\alpha)$ is an isomorphism, then $h(\alpha_1) = \beta_i$ for some $i\in \{1,\dots,d\}$.
\item Run a loop $i=1,\dots,d$ and for each $i$, use LLL[9] techniques to check if there exists a polynomial $h(x)\in \mathbb{Q}[x]_{<n}$ for which $h(\alpha_1)=\beta_i$. 
\end{enumerate}

\end{enumerate}

\par Method I is often fast, but one can give examples where it becomes slow, e.g. for so-called Swinnerton-Dyer polynomials, the degree $n^2$ factoring leads to a lattice reduction in [VH 2002] of dimension approximately $n^2/2$. Method I(b) can be faster, but one can still produce examples where it becomes slow, e.g. [7]. For such examples method II is faster because the lattice reduction there has dimension approximately $n$. 

\par Our algorithm is similar to Method II. There can be $n$ distinct isomorphisms (in the Galois case) and in this case our algorithm is the same as Algorithm II. However, if there is only one isomorphism then we can save roughly a factor $d$. This is because we can do the LLL computation for all $\beta_i$ simultaneously.

\section{Preliminaries} 

\par We let $\mathbb{Q}[x]_{<n}$ denote the polynomials over $\mathbb{Q}$ with degree less than $n$. 
\begin{defi} If $h(\alpha) \in \mathbb{Q}(\alpha)$ then the notation $h(x)$ is the element of $\mathbb{Q}[x]_{<n}$ that corresponds to $h(\alpha)$ under $x \mapsto \alpha$.
\end{defi}

\par Under an isomorphism $\beta$ will map to some $ h(\alpha) \in \mathbb{Q}(\alpha)$, \begin{equation} \beta \mapsto     h(\alpha) = \sum_{i=0}^{n-1} {h_i \alpha^i} \end{equation}.

\par A polynomial  $h(x) \in \mathbb{Q}[x]_{<n}$ represents an isomorphism if and only if $h(\alpha)$ is a root of $g$, i.e. $g(h(\alpha)) = 0$.

\par If $\mathbb{Q}(\beta)$ is isomorphic to $\mathbb{Q}(\alpha)$ then $g$ and $f$ have the same factorization pattern in $\mathbb{Q}_p[x]$ for every prime $p$. To simplify the factoring in $\mathbb{Q}_p[x]$ we restrict to good primes $p$, defined as:

\begin{defi} A good prime $p$ is one that does not divide the leading coefficient of $f$ or $g$ and does not divide the discriminant of either $f$ or $g$. 
\end{defi}
\begin{rmk} Both $f$ and $g$ are taken to be in $\mathbb{Z}[x]$. 
\end{rmk}
\par For a good prime $p$ we can factor $f$ in $\mathbb{Q}_p[x]$ up to any desired $p$-adic precision by factoring in $\mathbb{F}_p[x]$, followed by Hensel Lifting [11, p. 137]. Likewise we can distinct-degree factor $f$ as:

\begin{equation} \label{distdeg} f=F_1 F_2\dots F_m \text{  in  } \mathbb{Q}_p[x] \end{equation}
where $F_d$ is the product of all irreducible factors of $f$ in $\mathbb{Q}_p[x]$ of degree $d$ [11, Section 3.4.3].

\begin{defi} {\bf Sub-traces} Let $p$ be a prime and $d$ a positive integer. Then we define the $\mathbb{Q}$-linear map:
$$ Tr^d _p(f,*) : \mathbb{Q}(\alpha) \rightarrow \mathbb{Q}_p$$
as follows.
Let $h(x) \in \mathbb{Q}[x]_{<n}$, $h(\alpha) \in \mathbb{Q}(\alpha)$ and $F_d$ as above, then:
$$ Tr^d _p (f,h(\alpha)) := \sum _{\stackrel{\gamma \in \overline{\mathbb{Q}}_p}{F_d(\gamma)=0}} {h(\gamma)}  $$
We call these maps {\em sub-traces} because the sum is taken over a subset of the roots of $f$. \\
Likewise we define $Tr^d_p(g,*): \mathbb{Q}(\beta) \rightarrow \mathbb{Q}_p$. 
\end{defi} 

\begin{rmk} \label{trace} The map $Tr^d_p$ does not depend on the choice of the minimal polynomial $f$ that is used to represent the number field. In particular if $\beta \mapsto h(\alpha)$ is an isomorphism $\mathbb{Q}(\beta) \rightarrow \mathbb{Q}(\alpha)$ then $$Tr^d_p(g,\beta) = Tr^d_p (f,h(\alpha)) \text{   for every } p, d$$
\end{rmk}

\begin{defi} We will now define two bases of $\mathbb{Q}(\alpha)$ that we will need. The first one is the standard basis, which is $\{ 1, \alpha, \alpha^2, \dots, \alpha^{n-1} \}$. The second will be called the { \em rational representation basis}, which is \\ $\{1/f'(\alpha), \alpha/f'(\alpha), \dots, \alpha^{n-1}/f'(\alpha) \} $.
\end{defi}

\par Rational representation can improve running time and complexity results, see [4]. This representation has also been used under various names, see [2, 4], and occurs naturally in algebraic number theory as a dual basis under the trace operator, see [2].

\par A basis for $\mathbb{Q}(\alpha)$ corresponds to a map $\rho:  \mathbb{Q}^n \rightarrow \mathbb{Q}(\alpha)$. We use the rational representation basis, therefore $$\rho: (a_0, a_1, \dots, a_{n-1}) \mapsto \frac{1}{f'(\alpha)} \sum_{i=0}^{n-1} {a_i \alpha^i}.$$ 

\begin{defi} The inverse linear map $h(\alpha) \mapsto \vec{h}$, from $\mathbb{Q}(\alpha)$ to $\mathbb{Q}^n$ is as follows. Let $h(\alpha) = \sum^{n-1}_{i=0} {a_i \alpha^i} \in \mathbb{Q}(\alpha)$ and write $f'(\alpha) \cdot h(\alpha)$ as $\sum_{i=0}^{n-1} {b_i \alpha^i}$. Then define $\vec{h} := (b_0, b_1, \dots, b_{n-1}) \in \mathbb{Q}^n$. 
\end{defi}

\begin{rmk} \label{reason} One of the advantages of rational representation is: by using the $b_i$ in $\vec{h}$ instead of the $a_i$, we have $\vec{h} \in \mathbb{Z}^n $ for every algebraic integer $h(\alpha)$, see lemma~\ref{AlgInt}. Moreover, as in [4] this also improves bounds (section 4). It is also better to use $g_n \vec{h}$ than simply using $h(\alpha)$ since $g_n \vec{h}$ will have integer components, by Corollary~\ref{coralgint}, which are easier to bound and are heuristically of smaller size [4, Section 6].
\end{rmk}

For a polynomial $f(x)= \sum_{i=0}^n {f_i x^i}$ denote $$\|f(x)\| := \left( \sum_{i=0} ^n { |f_i|^2} \right) ^{1/2}.$$ 

\par Let $M(f)$ be the Mahler measure of $f$, $$M(f) := f_n \cdot \prod_{\stackrel{ f(\gamma)=0} { \gamma \in \mathbb{C}}} {\max{\{1, |\gamma| \}}}.$$

\section{Overview of the Algorithm}

{\bf Goal}: To find all $g_n \vec{h} \in \mathbb{Z}^n$ for which $\beta \mapsto h(\alpha)$ defines an isomorphism $\mathbb{Q}(\beta) \rightarrow \mathbb{Q}(\alpha)$. 
\\
{\bf Idea}: The aim of the pre-processing algorithm in Section 5 is to find a sequence $$ \mathbb{Z}^n = L_0 \supseteq L_1 \supseteq L_2 \supseteq \dots \supseteq L_k$$ such that all $g_n \vec{h}$ are in each $L_i$. We can then use $L_k$ to speed up the computation of the isomorphism(s), especially when $\text{dim}(L_k)$ is small. The cost of computing $L_k$ is comparable to one iteration in Method II.

\par In the algorithm we start with the lattice $\mathbb{Z}^n$ and then add the restrictions imposed by the condition that under an isomorphism, sub-traces $\mathbb{Q}(\alpha) \rightarrow \mathbb{Q}_p$ must correspond to sub-traces $\mathbb{Q}(\beta) \rightarrow \mathbb{Q}_p$. By doing this for several primes we are able to narrow down the possible isomorphisms. If dim($L_k$) $\leq 1$, this directly gives the isomorphism or shows that there is no isomorphism. If dim($L_k$) $>1$ then we switch to Method II, but starting with $L_k$. Thus we end up with $d$ lattice reductions of dimension dim($L_k$). In the worst case dim($L_k$) $ \approx n$, this costs the same as Method II. In the best case, dim($L_k$) $\leq 1$ and we save a factor $d$.

\section{Bounding the length of $g_n \vec{h}$}

\par To effectively carry out this algorithm we will need a good upper bound on the size of $g_n \vec{h}$. In this section we aim to find such a bound.

\begin{defi} Let $\alpha_1, \dots, \alpha_n \in \mathbb{C}$ be the roots of $f$. Then using the basis $\{1,x,x^2, \dots, x^{n-1}\}$ of $\mathbb{C}[x]_{<n}$  and the standard basis $\{e_1, e_2, \dots, e_n\}$ for $\mathbb{C}^n$, the interpolation map $\mathbb{C}^n \rightarrow \mathbb{C}[x]_{<n}$ is given by: $$e_i \mapsto  \frac{f(x)/(x-\alpha_i)}{f'(\alpha_i)}$$ 

This polynomial takes value $1$ at $x=\alpha_i$ and value $0$ at $x=\alpha_j$ ($i\neq j$). The inverse of the interpolation map is the evaluation map, which is given by the Vandermonde matrix: 

$$\begin{bmatrix} 1&\alpha_1&\alpha_1^2&\dots&\alpha_1^{n-1}\\ 1&\alpha_2&\alpha_2^2&\dots&\alpha_2^{n-1} \\ 1&\alpha_3&\alpha_3^2&\dots&\alpha_3^{n-1} \\ \vdots & \vdots &\ddots &\dots & \vdots \\ 1 & \alpha_n & \alpha_n^2 & \dots & \alpha_n^{n-1} \end{bmatrix}$$

\end{defi}

\begin{lemma} \label{AlgInt} If $a \in \mathbb{Q}(\alpha)$ is an algebraic integer and $f(x)$ is the minimal polynomial for $\alpha$, then $f'(\alpha)\cdot a \in \mathbb{Z}[\alpha]$. 
\end{lemma}
\begin{proof} Denote by $^{(i)}$ the complex embeddings of $\mathbb{Q}(\alpha)$. Then define 
$$ m(x) := \sum_{i=1}^n {a^{(i)} \frac{f(x)}{x-\alpha^{(i)}}}.$$ The coefficients of $m(x)$ are in $\mathbb{Q}$ since the polynomial is symmetric in the $\alpha^{(i)}$. But $m(x)$ is also a sum of polynomials all of whose entries are algebraic integers. Hence $m(x) \in \mathbb{Z}[x]$. Note that for $\alpha = \alpha^{(1)}$ we get $m(\alpha) = a f'(\alpha) \in \mathbb{Z}[\alpha]$. 
\end{proof}

\begin{cor} \label{coralgint} Let $\beta \mapsto h(\alpha)$ be an isomorphism of $\mathbb{Q}(\beta)$ and $\mathbb{Q}(\alpha)$. Then $g_n h(\alpha)$ is an algebraic integer and hence $g_n \vec{h} \in \mathbb{Z}^n$. 
\end{cor}

\begin{proof} Apply lemma~\ref{AlgInt} by letting $a = g_n h(\alpha)$ and recall that $g_n \vec{h}$ is comprised of the coefficients of $g_n f'(\alpha)h(\alpha)$ in the standard basis, each of which will be integers by lemma~\ref{AlgInt}.
\end{proof}

\begin{lemma} \label{inter} Let $P(x) = \sum_{i=1}^n{\beta_i \frac{f(x)}{x- \alpha_i}} \in \mathbb{Q}[x]_{<n}$, then $P(\alpha) = f'(\alpha) h(\alpha)$.
\end{lemma}
\begin{proof} If we evaluate $f'(x) h(x)$ at the roots of $f(x)$ and then interpolate we get: $$ \sum_{i=1}^n{\beta_i f'(\alpha_i) \frac{f(x)/(x-\alpha_i)}{f'(\alpha_i)} } = \sum_{i=1}^n{\beta_i \frac{f(x)}{x- \alpha_i}}.$$ Therefore $\sum_{i=1}^n{\beta_i \frac{f(x)}{x- \alpha_i}}$ will be the remainder of $f'(x) h(x)$ divided by $f(x)$, because they are of the same degree and coincide on the $n$ roots of $f(x)$. The lemma then follows from the fact that $\alpha$ is a root of $f(x)$.
\end{proof}

\par In order to bound $f'(\alpha)h(x)$ we will have to bound both $\frac{f(x)}{x-\alpha_i}$ and also $|\beta_i|$. We will use Corollary~\ref{grancor} to bound $\frac{f(x)}{x-\alpha_i}$ and since we know the $\beta_i$ up to a permutation (they are roots of $g(x)$), we can bound $\sum{|\beta_i|}$.

\begin{thm} \label{granthm} If $f(x)$ and $\tilde{f}(x)$ are polynomials with complex coefficients, of degree $n$ and $d$ respectively, such that $\tilde{f}(x)$ divides $f(x)$ and $|f(0)|=|\tilde{f}(0)| \neq 0$, then \begin{equation} \label{gran} \|\tilde{f}(x)\| \leq \left( \sum_{j=0}^{n-d} { { \binom {d}{j} } ^2}\right)^{1/2} \|f(x)\|. \end{equation}
\end{thm}

\begin{proof} See Granville, [1].
\end{proof}

\begin{cor} \label{grancor2} If $f(x)$ and $\tilde{f}(x)$ have the same leading coefficient and $f(0),\tilde{f}(0) \neq 0$ and $\tilde{f}(x)$ divides $f(x)$ then equation (\ref{gran}) holds.
\end{cor}
\begin{proof} Apply Theorem~\ref{granthm} to the reciprocals of $f$ and $\tilde{f}$.
\end{proof}

\begin{cor} \label{grancor} Let $P(x)$ be an irreducible  polynomial (over $\mathbb{Q}$) of degree $n \geq 1$ and let $\{\gamma_1, \gamma_2, \dots, \gamma_n\}$ be its complex roots. Then $$\left\|{\frac{P(x)}{x-\gamma_i}}\right\| \leq n \|P(x)\|$$
\end{cor}
\begin{proof}
Take $\tilde{f}(x)= {\frac{P(x)}{x-\gamma_i}}$ and $f(x)=P(x)$ and apply Corollary~\ref{grancor2}. Then $$\left\|{\frac{P(x)}{x-\gamma_i}}\right\| \leq \left( \sum_{j=0}^{n-(n-1)} { { {n-1} \choose j}^2}\right)^{1/2} \|P(x)\| = \left( \sum_{j=0}^{1} { { {n-1} \choose j} ^2} \right)^{1/2} \|P(x)\| $$ $$= \left( { {n-1} \choose 0}^2 + { {n-1} \choose 1}^2 \right) ^{1/2} \|P(x)\| = (n^2 -2n+2)^{1/2} \|P(x)\| \leq n \|P(x)\|.$$
\end{proof}

\begin{thm} \label{bound} Let $$S_{g(x)}:=\sum_{\stackrel{g(\beta)=0}{\beta \in \mathbb{C}}}{|\beta_i|}$$ then: \begin{equation}  g_n \|\vec{h} \| \leq g_n n \left(S_{g(x)} \right) \|f(x)\| .\end{equation}
There are several ways to bound $S_{g(x)}$:\\
1) $S_{g(x)} \leq$ The degree of $g(x)$ times the rootbound described in [3].\\
2) $S_{g(x)} \leq M(g)/lc(g) + (n-1)$, where the Mahler measure can be bounded by $\|g(x)\|$.

\end{thm}

\begin{proof} (of Equation (4)) $$ \| \vec{h} \| = \| P \| = \| \sum_{i=1}^n{\beta_i \frac{f(x)}{x- \alpha_i}} \| \leq   n \|f(x)\| \sum_{i=1}^n{|\beta_i}| =  n \|f(x)\| S_{g(x)} .$$ The first equality is by the definition of $\vec{h}$, the second by Lemma~\ref{inter} and the inequality by Corollary~\ref{grancor}.
\end{proof}

\section{The Algorithms}

Here we give the algorithms for computing the isomorphisms between number fields. The {\em Pre-processing} algorithm reduces the lattice of possible isomorphisms and gives the explicit isomorphism if there is only one. The next algorithm, {\em FindIsomorphism}, calls the {\em Pre-processing} algorithm and uses the remaining lattice to check which maps on roots corresponds to an isomorphism. 

{\bf Algorithm: LLL-with-removals}[10] \\ \itshape
{\bf Input} A matrix $A$ and a bound $b$. \\
{\bf Output} A set of LLL reduced row vectors where the last vector is removed if its Gram-Schmitt length is greater than $b$.\\

{\bf Algorithm: FindSuitablePrime} \\ \itshape
{\bf Input} $(f(x), g(x), x, {\rm bp}, b, e)$, where {\rm bp} is the first prime to test, $b$ and $e$ determine the level to Hensel Lift to.  \\
{\bf Output} $p, p^a, m, [ [F_{d_1},G_{d_1}], [F_{d_2},G_{d_2}], \dots [F_{d_m},G_{d_m}]]$, see equation ({\rm 2}) for notation.\\
{\bf Procedure} \\

\begin{enumerate}
\item $p:= {\rm bp}$, counter$:=0$.
\item Repeat (until the algorithm stops in Steps 2(d)ii, 2(f) or 2(j)).
\begin{enumerate}
	\item $p:= \text{nextprime}(p)$
	\item if $p |$ discriminant($f,x$) or $p | f_n$ then go to Step 2(a)
	\item if $p |$ discriminant($g,x$) or $p | g_n$ then go to Step 2(a)
	\item Distinct Degree Factor $f$ as $f \equiv F_{d_1} F_{d_2} \dots F_{d_m} \mod p$. 
	\begin{enumerate}
		\item If $m=1$ then counter := counter $+1$.
		\item If counter $>25$ then print ``{\rm They appear to be Galois}'' and return 0,0,0,0.
		\item Return to Step 2(a).
	\end{enumerate}
	\item Distinct Degree Factor $g$ as $g \equiv G_{d'_1} G_{d'_2} \dots G_{d'_{m'}} \mod p$. 
	\item If $m \neq m'$ or if the degrees of $F_i$ and $G_i$ do not match then return ``{\rm There is no isomorphism}''.
	\item Let $a := \left\lceil b^{e/10} 2^{e/4} \right\rceil$.
	\item Hensel lift $f \equiv F_{d_1} F_{d_2} \dots F_{d_m} \mod p^a$ and likewise for $g$.
	\item If deg$(F_1)>0$ then store $p$ for later use.
	\item Return $p, p^a, m, [ [F_{d_1},G_{d_1}], [F_{d_2},G_{d_2}], \dots [F_{d_m},G_{d_m}]]$ as output and stop.
	\end{enumerate}
\end{enumerate}

{\bf Algorithm: Pre-Processing} \\ \itshape
{\bf Input} Polynomials $f(x)$ and $g(x)$. \\
{\bf Output} Either ``No isomorphism exists", a verified isomorphism, or a $\mathbb{Z}$-module which contains $(g_n \vec{h}, g_n)$ for every isomorphism $h$.\\
{\bf Remark}: This lattice is given as the row space of a matrix $C$.\\
{\bf Procedure} \\

\begin{enumerate}
\item \emph{ Initialize}
\begin{enumerate}
\item $e:=n+1$.
\item $C :=$ $(n+1)$ {\rm x} $(n+1)$ identity matrix. 
\item $p:=3$.
\item $q:=0$.
\item Let $\{ {\rm Base}_i\} \in \mathbb{Q}(\alpha)_{<n}, i=1\dots n$ be $ \{ \rho(1, 0, \dots, 0), \rho(0, 1, \dots, 0), \dots, \rho(0, 0, \dots, 1) \}$ with $\rho$ defined in Section 2.
\end{enumerate}
 
\item Let $S$ be an upper bound for $\sum_{\stackrel{g(\beta)=0}{\beta \in \mathbb{C}}}{|\beta_i|}$, e.g. (\ref{bound}.1) or (\ref{bound}.2). Our implementation uses (\ref{bound}.1).
\item Let $b := n S \|f(x)\|$, be the bound described in Theorem~\ref{bound}.
\item Repeat (until the algorithm stops in 4(b), 4(e) or 4(i)).
\begin{enumerate} 
\item $q:= q+1$.
\item $p, p^a, m, M_q :=$ {\rm FindSuitablePrime}$(f, g, x, p, b, e)$.
\begin{enumerate}
\item If $p=0$ then return $C$.
\end{enumerate}
\item Find $Tr_p^d(f,{\rm Base}_i)$ for $i =1 \dots n$ and $Tr_p^d(g,\beta)$ for each $d$ with deg($F_d$)$ >0$. The necessary $F_d,G_d$ are read from $M_q$.

\item $A:=\left[ \begin{array}{cc} C & CT \\ 0 & P \end{array} \right] $, where 
$$P:= \left[ \begin{array}{ccc} p^a&& \\ &\ddots& \\ &&p^a \end{array} \right],$$ 
$$T:= \left[\begin{array}{ccc}
 Tr^{d_1}_{p_1}(f, {\rm Base}_1)&\dots&Tr^{d_m}_{p_1}(f, {\rm Base}_1) \\ 
 Tr^{d_1}_{p_1}(f, {\rm Base}_2)&\dots&Tr^{d_m}_{p_1}(f, {\rm Base}_2)\\
  \vdots&\dots&\vdots \\
  Tr^{d_1}_{p_1}(f, {\rm Base}_n)&\dots&Tr^{d_m}_{p_1}(f, {\rm Base}_n)\\
  Tr^{d_1}_{p_1}(g,\beta)&\dots& Tr^{d_m}_{p_1}(g,\beta) 
\end{array} \right]$$
 the $d_1,\dots,d_m$ are as in Step 2(h) in Algorithm FindSuitablePrime. (Omitted entries are zero.)
\item If $CT \equiv 0 \mod p^a$ then 
 \begin{enumerate}
	\item counter $:=$ counter $+1$.
	\item If counter $< 10$ then {\rm Go to Step 4(a)} else return $C$ and stop.
 \end{enumerate}

\item $L:=$ LLL-with-removals($A$, $b$). 

\item Let $C$ be the matrix with the first $n+1$ columns of $L$ and $B$ the remaining $m$ columns of $L$, so $L=[\begin{array}{cc} C&B \end{array}]$.

\item if $B \neq 0$ then

 \begin{enumerate}
	\item $ B:= 10^{20} \cdot B$ 
	\item $A := [\begin{array}{cc} C&B \end{array}]$
	\item $L := \text{LLL-with-removals}(A,b)$, then go to Step 4(g).
 \end{enumerate}
\item Let $e := $ number of rows of $C$.
	\begin{enumerate}
	\item if $e=0$ then output ``There is no isomorphism."
	\item if $e=1$ then let $C$ be $[V,v]$ with $V$ an $n$ dimensional vector, and let $h$ be the polynomial corresponding to $V/v$.
		\begin{enumerate}
		\item Let {\rm iso}:= $\frac{h(\alpha) g_n}{f'(\alpha)}$.
		\item If {\rm iso} satisfies $g$ then output ``{\rm iso} is the only isomorphism."
		\item If not then output ``There is no isomorphism."
 		\end{enumerate}
	\item Else, go to Step 4a.
	\end{enumerate}
\end{enumerate}
\end{enumerate}
\normalfont
\begin{rmk}  If we let $d$ be the number of isomorphisms ($\mathbb{Q}(\beta) \rightarrow \mathbb{Q}(\alpha)$) then just by looking at the input/output of the Pre-processing algorithm we see that:
$$ \text{If } d \in \{0,1\} \text{ then the output is either } \left\{ \begin{array}{c} \text{all isomorphisms} \\ \text{a lattice} \end{array} \right.   $$
$$ \text{If } d > 1 \text{ then the output is a lattice} $$
\end{rmk}

\par  In the next algorithm we use the lattice outputted from {\em Pre-Processing} to check all possible maps on the roots to see which are actual isomorphisms. This will find all isomorphisms from $\mathbb{Q}(\beta) \rightarrow \mathbb{Q}(\alpha)$.

\par The following algorithm is described for (linear) roots of $f$ and $g$ in $\mathbb{Q}_p$ and can be extended to the roots of $F_i$ and $G_i$ instead. 

\begin{rmk} \label{useful} It should be noted that even if the Pre-processing Algorithm does not find the isomorphism(s), the LLL switches it performs will still contribute to the FindIsomorphism Algorithm. This is true for the same reason as in [11, pg 175].
\end{rmk}

{\bf Algorithm: FindIsomorphism} \\ \itshape
{\bf Input} Two polynomials, $f,g \in \mathbb{Z}[x]$ which are irreducible and of the same degree. \\
{\bf Output} The set of all isomorphisms from $\mathbb{Q}[x]/(f)$ to $\mathbb{Q}[x]/(g)$. \\
{\bf Procedure} \\

\begin{enumerate}
\item $C$ := {\em Pre-Processing}($f(x)$ , $g(x)$, $x$).
\item If Step 2(i) in Algorithm {\em FindSuitablePrime} (called from {\em Pre-Processing}) stored at least one prime, then choose one with smallest $\text{deg}(F_1)$. Otherwise keep calling Algorithm {\em FindSuitablePrime} until such a prime is found.
\item Let $\alpha_1, \dots, \alpha_d$ be the roots of $F_1$ and Hensel lift them to $\mathbb{Z}/(p^a)$ with $a$ as in Algorithm {\em FindSuitablePrime}. Likewise let $\beta_1, \dots, \beta_d \in \mathbb{Z}/(p^a)$ be the roots of $G_1$.
\item For $j$ from 1 to $d$ do: 
  \begin{enumerate}
	\item Apply steps 4(d) through 4(i)ii of {\em Pre-Processing} using 
	$$T:= \left[\begin{array}{c}
 {\rm Base}_1|_{\alpha = \alpha_j} \\ 
 {\rm Base}_2|_{\alpha = \alpha_j}\\
  \vdots \\
 {\rm Base}_n|_{\alpha = \alpha_j}\\
 \beta_1 
\end{array} \right]$$ 
	\item If $e > 1$ then 
	\begin{enumerate}
			\item Hensel Lift the roots of $f$ and $g$ to twice the current $p$-adic precision, i.e. $p^{2a}$.
			\item Apply Step 4(a) with the more precise roots. 
	\end{enumerate}
\end{enumerate}

\end{enumerate}
\normalfont

\subsection{Proofs of Termination and Validity}
In this section we prove that the algorithms terminate and show that the algorithm does indeed produce all isomorphisms of the number fields $\mathbb{Q}(\beta)$ and $\mathbb{Q}(\alpha)$.

\par First we cite a lemma which shows why we can use LLL with removal in our algorithm. 

\begin{lemma} \label{LLL} Let $\{ b_1, \dots, b_k \}$ be a basis for a lattice, $C$, and $\{ b_1^*, \dots , b_k^* \}$ the corresponding Gram-Schmitt orthogonalized basis for $C$. If $ \| b_k^* \| > B$ then a vector in $C$ with norm less than $B$ will be a $\mathbb{Z}$-linear combination of $\{ b_1, \dots, b_{k-1} \}$.
\end{lemma}
\begin{proof} This follows from the proof of Proposition 1.11 in [9], it is also stated as Lemma 2 in [10].
\end{proof} 

\begin{cor} Using LLL-with-removals on a lattice containing $g_n\vec{h}$ with the bound $b$, computed in Step 3 of {\em Pre-Processing}, does not remove $g_n\vec{h}$ from the lattice.
\end{cor}
\begin{proof} Using Lemma~\ref{LLL} and Theorem~\ref{bound} we know that removing final vectors with Gram-Schmitt length bigger than $b$ does not remove any of the $g_n \vec{h}$.
\end{proof}

\begin{lemma} The {\em Pre-Processing} Algorithm terminates.
\end{lemma}
\begin{proof} The steps of the {\em Pre-Processing} algorithm are known algorithms that terminate, the only one that is not immediate is Step 4(h). Step 4(h) terminates because each run increases the determinant of the lattice (Step 4(h)i) and any final (see Lemma~\ref{LLL}) vector with Gram-Schmitt length bigger than $b$ is removed, thus the number of vectors is monotonically decreasing and hence it can only be run a finite number of times. 
\end{proof}

\begin{lemma} The {\em FindIsomorphism} Algorithm described above terminates. 
\end{lemma}
\begin{proof} 
\par For Steps 1-3 it is clear why each will terminate. We show that Step 4 terminates by contradiction. \\
Suppose Step 4 never terminates (i.e. the lattice always has dimension $>$ 1) then it contains at least two vectors: $(h_1, e_1)$ and $(h_2,e_2)$. Let $H=h_1$ if $e_1=0$ or $H=e_1 h_2- e_2 h_1$ otherwise. Then $H(\alpha) \equiv 0 \mod p^a$. We get a contradiction when $p^a$ is larger than an upper bound for ${\rm Res}_x(H,f)$. An upper bound for $H$ can be obtained from equation $1.7$ in [9] and the fact that the last vector after LLL-with-removals has Gram-Schmitt length $\leq b$.
\end{proof}

\section{Heuristic estimate on the rank of $C$}

Let $C \subseteq \mathbb{Z}^{n+1}$ be the output of the {\rm Pre-Processing} Algorithm. 

\begin{observation} In most (but not all) examples, dim($C$) is equal to $n- n/d + 1$.
\end{observation}

This means that {\rm Pre-Processing} is most effective when $d=1$.  Though as pointed out in Remark~\ref{useful} the work done in {\rm Pre-Processing} reduces the amount left to do.

\par Let $G$ be the Galois group of  $f(x)$ and let $H_i$ be the stabilizer of $\alpha_i$ for $i\in \{1,2,\dots,n\}$, where the $\alpha_i$ are the roots of $f(x)$. 
\par Let $d$ be the number of $j$ such that $H_1 = H_j$, then $d$ is the number of automorphisms of $\mathbb{Q}(\alpha)$. If $\mathbb{Q}(\alpha)$ and $\mathbb{Q}(\beta)$ are isomorphic then $d$ will also be the number of isomorphisms from $\mathbb{Q}(\beta)$ to $\mathbb{Q}(\alpha)$.

\begin{rmk} We view $G$, which as the Galois group acts on $\{\alpha_1,\alpha_2,\dots,\alpha_n\}$, as acting on the set $\{1,2,\dots,n\}$ in the most natural way. Hence we view $G$ as a subgroup of $S_n$, the symmetric group.
\end{rmk}

We will construct a partition matrix as follows. For each $\sigma \in G$, group together the cycles of the same length. Different group elements and cycle lengths will correspond to different rows. For each element of $G$ and for each cycle length in $\sigma$, construct one row of $P$ as follows: place a $1$ in the $i^{th}$ entry if $\alpha_i$ is in a cycle of that length. We call the resulting matrix $P$.

\par For example for $\sigma_1 = (1)(2)(3)(456)$ and $\sigma_2 = (12)(3456)$ we would get the following partition matrix :
$$ P=\begin{array}{c} \sigma_1 \text{   } l=1\\ \sigma_1 \text{   } l=3\\ \sigma_2 \text{   } l=2\\ \sigma_2 \text{   } l=4\\ \vdots \end{array}\begin{bmatrix} 1&1&1&0&0&0 \\ 0&0&0&1&1&1 \\ 1&1&0&0&0&0 \\ 0&0&1&1&1&1 \\ \vdots&\vdots&\vdots&\vdots&\vdots&\vdots \end{bmatrix}$$

\par Since there are $d$ automorphisms the number of distinct columns of $P$ will be $\leq n/d$, hence $\text{rank}(P) \leq n/d$ and thus $\text{Nullspace}(P) \geq n-n/d.$

\par This translates into an estimate on the rank of the lattice $C$ since it helps us bound $$V = \bigcap_{p,d} {Ker(Tr_p^d(f,*))}.$$ 
\par Nullspace($P$) corresponds to elements for which all sub-traces are zero, so dim(Nullspace($P$)) $\leq$ dim($V$). 
\par Since we used LLL-with-removals with cut off point $b$, if $V$ admits a basis whose norms are all smaller than $b$ then $V \subseteq \pi_{1\dots n}(C)$, where $\pi_{1\dots n}$ is the projection on the first $n$ coordinates. 

Therefore under that assumption $$\text{dim}(\pi_{1\dots n}(C) ) \geq \text{dim(Nullspace($P$))} \geq n- n/d.$$ 

This leads to our estimate: \begin{equation} \label{rank} \text{dim}(C) \approx n-n/d +1. \end{equation}

\par For most polynomials taken from the database [5] our estimate is an equality. Peter Muller provided an infinite sequence of counter-examples for the case we were most interested in ($d=1$). For the first group in this sequence, the database [5] provides the following example: \\
$f := x^{14}+2x^{13}-5x^{12}-184x^{11}-314x^{10}+474x^9+1760x^8+1504x^7-400x^6-1478x^5-818x^4+73x^3+260x^2+121x+23,$ \\
which has one automorphism but the Pre-processing algorithm outputs a dimension $2$ lattice.

\section{Computational Efficiency}
We compare our algorithm implemented in Maple with other methods of finding isomorphisms. The best algorithm we know for factoring over number fields is give by Belabas in [6], which is implemented in Pari/Gp. Both implementations were run on a standard 2.2GHz processor. We tested them on the field extensions given by the following two degree 25 polynomials: \\
$f_1 := 2174026154062500000\,{x}^{25}-12927273797812500000\,{x}^{24}+
44254465332187500000\,{x}^{23}-102418940816662500000\,{x}^{22}+
180537842164766250000\,{x}^{21}-249634002590534050000\,{x}^{20}+
292282923494920350000\,{x}^{19}-384197583430502150000\,{x}^{18}+
815826517614521346000\,{x}^{17}-2131245874043847615600\,{x}^{16}+
4352260622811059705104\,{x}^{15}-6463590834754261173232\,{x}^{14}+
6920777688226436002712\,{x}^{13}-4525061881234027826296\,{x}^{12}+
528408698276686662696\,{x}^{11}+2762117617850418790424\,{x}^{10}-
4343360968383689825174\,{x}^{9}+4191186502263628451150\,{x}^{8}-
2802452375464033976482\,{x}^{7}+1332292171242725153638\,{x}^{6}-
161285249796825311495\,{x}^{5}-429207332210687640181\,{x}^{4}+
264147194777000152867\,{x}^{3}+6032198632961699729\,{x}^{2}-
42885793067858008650\,x+13774402803823804220$ and \\ 
$f_2 := -42885793067858008650\,x-13774402803823804220-161285249796825311495\,{
x}^{5}+429207332210687640181\,{x}^{4}+264147194777000152867\,{x}^{3}-
6032198632961699729\,{x}^{2}-1332292171242725153638\,{x}^{6}-
2802452375464033976482\,{x}^{7}-4191186502263628451150\,{x}^{8}-
4343360968383689825174\,{x}^{9}-2762117617850418790424\,{x}^{10}+
528408698276686662696\,{x}^{11}+4525061881234027826296\,{x}^{12}+
6920777688226436002712\,{x}^{13}+6463590834754261173232\,{x}^{14}+
4352260622811059705104\,{x}^{15}+815826517614521346000\,{x}^{17}+
384197583430502150000\,{x}^{18}+292282923494920350000\,{x}^{19}+
249634002590534050000\,{x}^{20}+180537842164766250000\,{x}^{21}+
102418940816662500000\,{x}^{22}+44254465332187500000\,{x}^{23}+
2131245874043847615600\,{x}^{16}+12927273797812500000\,{x}^{24}+
2174026154062500000\,{x}^{25}.$ 

\par These are field extensions with one isomorphism between them. Using Belabas' method we have a runtime of 11.69 seconds, which includes the indispensable operation of defining the number field, and with our algorithm we have a runtime of 2.97 seconds.

\par We also tested this algorithm in the case where our heuristic estimate on the rank does not apply, namely on the first counter-example. Using Belabas' method we have a runtime of .091 seconds and with our algorithm we have a runtime of .797 seconds. 

\par Bearing in mind that our implementation is not optimized and is coded in Maple we expect it to be much faster when using a better implementation of LLL. 

\par We also tested them on a larger example, namely the degree 81 example located at [7], our algorithm found the isomorphism in 2326.581 seconds and the Belabas' implementation did not finish as it ran out of memory after trying for a few days. Therefore there are certainly advantages to using this algorithm, as there are examples where there is a significant reduction in the required runtime/resources.

\section{Summary}
Method II (from Section 1) can be described by the following procedure: first pick $p$ such that $f$ and $g$ have roots in $\mathbb{Q}_p$. Fix one root $\beta \in \mathbb{Q}_p$ of $g$, take all roots $\alpha_1, \dots, \alpha_d \in \mathbb{Q}_p$ of $f$. Then for each $\alpha_i$ use LLL to find $h_i \in \mathbb{Q}[x]$ (if it exists) with $h_i(\alpha_i) = \beta_i$.

\par Our approach is similar, the difference is that we start with LLL reductions (obtained from sub-traces) that are valid for all $\alpha_i$. This way, a portion of the LLL computation to be done for each $\alpha_i$ is now shared. The time saved is then $(d-1)$ times the cost of the shared portion. This can be made rigorous by introducing a progress counter for LLL cost similar to [10].
\section{References}

\begin{enumerate}

\item Granville, A. ``Bounding the coefficients of a divisor of a given polynomial", Monatsh. Math. 109 (1990), 271-277. 
\item Conrad, Kieth. ``The different ideal". Expository papers/Lecture notes. Available at: \\ http://www.math.uconn.edu/$\sim$kconrad/blurbs/gradnumthy/different.pdf
\item Monagan, M. B. ``A Heuristic Irreducibility Test for Univariate Polynomials", J. of Symbolic Comp., 13, No. 1, Academic Press (1992) 47-57.
\item Dahan, X. and Schost, $\acute{E}$. 2004. ``Sharp estimates for triangular sets". In Proceedings of the 2004 international Symposium on Symbolic and Algebraic Computation (Santander, Spain, July 04 - 07, 2004). ISSAC '04. ACM, New York, NY, 103-110. 
\item Database by J$\ddot{u}$rgen Kl$\ddot{u}$ners and Gunter Malle , located at: \\ http://www.math.uni-duesseldorf.de/$\sim$klueners/minimum/minimum.html
\item Belabas, Karim. ``A relative van Hoeij algorithm over number fields". J. Symbolic Computation, Vol. 37 (2004), no. 5, pp. 641-668.
\item Website with implementations and Degree 81 examples:\\ http://www.math.fsu.edu/\  $\sim$vpal/Iso/

\item van Hoeij, Mark. ``Factoring Polynomials and the Knapsack Problem." J. Number Th. 95, 167-189, 2002
\item Lenstra, A. K.; Lenstra, H. W., Jr.; Lov\'{a}sz, L.  ``Factoring polynomials with rational coefficients". Mathematische Annalen 261 (4), 515-534, 1982.
\item M. van Hoeij and A. Novocin, `` Gradual sub-lattice reduction and a new complexity for factoring polynomials", accepted for proceedings of LATIN 2010.
\item Cohen, Henri A Course in Computational Algebraic Number Theory, Graduate Texts in Mathematics 138, Springer-Verlag, 1993.
\end{enumerate}

Florida State University 211 Love Building, Tallahassee, Fl 32306-3027, USA\\
{\em E-mail address}: hoeij@math.fsu.edu\\
{\em E-mail address}: vpal@math.fsu.edu\\

\end{document}